\documentclass{ragtime}

\usepackage{amsthm}

\newtheorem{prop}{Proposition}

\newcommand{\be}{\begin{equation}}
\newcommand{\ee}{\end{equation}}
\newcommand{\bea}{\begin{eqnarray}}
\newcommand{\eea}{\end{eqnarray}}

\title[Probing Dark Energy through Perfect Fluid Thermodynamics]
{Probing Dark Energy through Perfect Fluid Thermodynamics}

\author[G. Lukes-Gerakopoulos, G. Acquaviva  % Run. head authors: separate names with commas,
     	and C. Markakis	]%    Now let's start the paper title authors:
       {Georgios Lukes-Gerakopoulos\at[]{1,a},
       Giovanni Acquaviva\at[]{2,b} % Makes referencing superscript `1'
        \splitauthors and Charalampos Markakis\at[]{3,4,5,c} % ref. superscr. `1,a',
       \\ % Termination of authors' block; if
        \ins{1}Astronomical Institute of the Academy of Sciences of the Czech Republic,\splitins[1]% This is how to break an
        Bo\v{c}n\'{i} II 1401/1a, CZ-141 31 Prague, Czech Republic\\% Termination of the first affiliation.
        \ins{2} Institute of Theoretical Physics, Faculty of Mathematics and Physics,
        \splitins[2] Charles University, CZ-180 00 Prague, Czech Republic\\
        \ins{3} DAMTP, University of Cambridge, Wilberforce Rd, Cambridge CB3 0WA, UK \\
        \ins{4} NCSA, University of Illinois at Urbana-Champaign, 
        \splitins[4] 1205 West Clark St, MC-257, Urbana, IL 61801, USA \\
        \ins{5} School of Mathematical Sciences, Queen Mary University of London,
        \splitins[5]  Mile End Road, London E1 4NS, UK 
\\% Termination of the second affiliation.
        \ins{a}\Email{gglukes@gmail.com} % This is how to present E-mail.
        \ins{b}\Email{gioacqua@gmail.com}
        \ins{c}\Email{c.markakis@damtp.cam.ac.uk} } % This is how to present E-mail.

\begin{document}

\begin{abstract}
    We demonstrate that the thermodynamics of a perfect fluid describing baryonic matter can, in certain limits, lead to an equation of state similar to that of dark energy. We keep the cosmic fluid equation of state quite general by just demanding that the speed of sound is positive and less than the speed of light. In this framework, we discuss some propositions by looking at the asymptotic behaviour of the cosmic fluid.
\end{abstract}

\begin{keywords}
dark energy--perfect fluid thermodynamics
\end{keywords}

\section{Introduction}

In this work we attempt to tackle the issue of dark energy \citep[see, e.g.,][]{Peebles:2003:RMP:}  by considering just usual baryonic matter in an ever-expanding Universe. We try to keep the investigation's assumptions as general as possible. Thus, we do not specify the equation of state (EOS) and we avoid to limit the study to a specific spacetime. In this framework the baryonic matter is described by an irrotational relativistic perfect fluid. For our analysis we follow a perfect fluid formalism introduced by \citet{Lichnerowicz:1967:RHM:} and \citet{Carter:1979:AGN:}, which in recent works was employed mainly for neutron stars \citep[see, e.g.,][]{Gourgoulhon:2006:EAS:,Markakis:2016:PRD:}. 

In particular, we consider a perfect fluid in an equilibrium configuration with proper energy
density $\epsilon$.  The state of the fluid depends on two parameters, which can
be taken to be the rest-mass density $\rho$ and specific entropy (entropy per unit rest-mass) $s$.
Then the EOS of the fluid is given by a function 
\begin{align}  \label{eq:EOSi}
 \epsilon  = \epsilon (\rho,s )\, .
\end{align}
From Eq.~\eqref{eq:EOSi} one can derive the first law of thermodynamics:
\begin{align}\label{eq:1stlaw}
    d\epsilon= \mu \frac{d \rho}{m_b} + T d(s \rho)\, ,
\end{align}
where $m_{\rm{b}}$ denotes the rest mass of a baryon and $\mu$ is the baryon chemical potential.
The pressure $p$ and specific enthalpy
$h$ are functions of $\rho$ and $s$  entirely determined by Eq.~\eqref{eq:EOSi}:
\begin{align} 
 p &=  - \epsilon  + \rho\, T\, s  + \frac{\mu}{m_{\rm{b}}}\rho \label{eq:pEOS}\, , \\
 h &:= \frac{{\epsilon  + p}}{\rho } = \frac{\mu }{{{m_{\rm{b}}}}} + Ts\, \label{eq:enth}\, .
\end{align}
Note that Eq.~\eqref{eq:pEOS} can be obtained by the extensivity property of the energy density, while the second equality of Eq.~\eqref{eq:enth} comes from Eq.~\eqref{eq:pEOS}.  Now Eqs.~\eqref{eq:1stlaw} and \eqref{eq:enth} yield the thermodynamic relations 
\begin{align}
d\epsilon  &= h\, d\rho  + \rho\, T\, ds \label{eq:de}\, , \\
dp &= \rho\, dh - \rho\, T\, ds   \label{eq:dp}\, .
\end{align}
Moreover, writing $h = h(\rho ,s)$   and differentiating yields
\be  \label{eq:dh}
dh = \frac{{hc_{\rm{s}}^2}}{\rho }d\rho  + {\left. {\frac{{\partial h}}{{\partial s}}} \right|_{\rho}}ds\, ,
\ee
where 
\be   \label{eq:cs2}
c_{\rm{s}}^2 = {\left. {\frac{{\partial p}}{{\partial \epsilon }}} \right|_s} = \frac{\rho }{h}{\left. {\frac{{\partial h}}{{\partial \rho }}} \right|_s}
\ee
is the sound speed.  In order to ensure causal evolution, given the upper bound
for signal propagation set by the speed of light, physically admissible fluids
should have
\begin{align} \label{eq:SpSoL}
0\lesssim s_m^2\leq c_s^2\leq 1,
\end{align}
where $s_m^2$ is an arbitrarily close to zero cut-off value for the speed of sound.

A simple perfect fluid is characterized by the energy-momentum tensor
\begin{align}
 {T_{\alpha}}^{\beta} = h\, \rho\, u_\alpha u^\beta + p\, {g_{\alpha}}^{\beta}
= \left( \epsilon + p \right)\, u_\alpha u^\beta + p\, {g_{\alpha}}^{\beta}\ ,
\end{align}
where $ {g_{\alpha \beta }}$ is the spacetime metric and $u^{\mu}$ is the 
timelike vector tangent to the fluid's flow, satisfying the normalization
condition ${u^\alpha }{u_\alpha } =  - 1$.  Such energy-momentum tensor is the 
source in Einstein's field equations (EFE) ${G_{\alpha}}^{\beta} = {T_{\alpha}}^{\beta}\,$, which are assumed to hold throughout this work.
By taking the covariant divergence of EFE, the doubly contracted Bianchi
identities $\nabla_{\beta}{G_{\alpha}}^{\beta}\equiv0$ assure the covariant 
conservation of energy-momentum
\be
\nabla_{\beta} {T_{\alpha}}^{\beta}=0\, ,\label{emcons}
\ee
which is the relativistic version of Euler equation.  Using Eq.~\eqref{eq:dp} with variation evaluated along the flow lines
($d \rightarrow u^{\alpha}\nabla_{\alpha}$) and thanks to the normalization of
the timelike vector $u^{\alpha}$, eq.\eqref{emcons} takes the form
\be \label{ValenciaT}
  \nabla_\alpha T^\alpha_{\,\,\,\,\beta}  =p_\beta\nabla_\alpha (\rho u^\alpha)+
  \rho[ u^{\alpha} \Omega_{\alpha \beta}- T \nabla_\beta  s ]= 0\, ,
\ee
where $ {p_\alpha } = h{u_\alpha }$ is the {\it canonical momentum} of a fluid element, and its exterior derivative $ \Omega_{\alpha \beta}:=\nabla_\alpha {\rm{}} p_\beta-\nabla_\beta {\rm{}} p_\alpha\,$
is the {\it canonical vorticity 2-form}.  If we assume the rest-mass (or baryon) conservation
\be   \label{eq:continuityeqn}
\nabla_\alpha (\rho u^\alpha)=0\, ,
\ee
eq.~\eqref{ValenciaT} yields the relativistic Euler
equation in  the canonical form:
\be  \label{eq:EulerCanonical}
    u^{\alpha} \Omega_{\alpha \beta}= T \nabla_\beta s\, .
\ee
Contraction of eq.~\eqref{eq:EulerCanonical} with the four-velocity
${u^\beta }$ makes the left-hand side vanish identically.\footnote{This is
because the left-hand side, after contraction with $u^{\beta}$, ends up being a
product of the symmetric term $u^{\alpha}u^{\beta}$ with the antisymmetric 2-form
$\Omega_{\alpha \beta}$.}  Hence the specific entropy is constant along the flow lines:
\be \label{eq:adiabatic}
{u^\alpha }{\nabla _\alpha }s = 0\, .
\ee
This reflects the fact that the Euler equation describes {\it adiabatic flows},
{\it i.e.} there are no heat fluxes in the fluid
nor particle production.  The adiabatic character of the fluid as expressed by
Eq.~\eqref{eq:adiabatic} is a consequence of assuming rest-mass conservation
Eq.~\eqref{eq:continuityeqn}.  

\section{Thermodynamical relations for an irrotational fluid} \label{sec:ThermoRel}

The condition for irrotational fluid flow is $\Omega_{\alpha \beta}=0$,
and  implies through Eq.~\eqref{eq:EulerCanonical} that the specific entropy is
constant, {\it i.e.} $\text{d} s=0$. The fundamental relations Eqs.~\eqref{eq:de}-\eqref{eq:dh} reduce to
\begin{align}
 d h &=\frac{h\, c_{\rm{s}}^2 }{\rho}\, d\rho\, \label{eq:drho} ,\\
 d \epsilon &=h\, d\rho\label{eq:depsilon}\, ,\\
 d p &=\rho\, dh\, . \label{eq:dpdh}
\end{align}
Using the limits set by Eq.~\eqref{eq:SpSoL} and making the reasonable assumption 
that the rest-mass density is a positive quantity, since we consider fluid
composed only of baryonic matter, we arrive through Eq.~\eqref{eq:drho} to 
\begin{align} \label{eq:drhoInt}
 \int_{\rho_1}^{~\rho} \frac{s_m^2 d\rho'}{\rho'} \le
 \int_{\rho_1}^{~\rho} \frac{c_s^2 d\rho'}{\rho'} \le \int_{\rho_1}^{~\rho} \frac{d\rho'}{\rho'} \quad
 \Rightarrow \quad \left(\frac{\rho}{\rho_1}\right)^{s_m^2} \le\frac{h}{h_1}\le \frac{\rho}{\rho_1}\, ,
\end{align}
where index ``1'' refers to the integration constants of the specific fluid with 
equation of state described by the speed of sound $c_s^2$, not by the lower
bound and upper bounds of Eq.~\eqref{eq:SpSoL}. Note that we have assumed that
$d\rho>0$. Eq.~\eqref{eq:drhoInt} implies
$ \displaystyle \left(\frac{\rho}{\rho_1}\right)^{s_m^2-1}\le 1$ ,
which gives that $\rho_1\le\rho$, since  $s_m^2<1$, i.e. the integration constant
$\rho_1$ corresponds to the minimum of the allowed values for the rest-mass
density of the fluid. Moreover, inequality~\eqref{eq:drhoInt} implies that
$h/h_1>0$. At this point we do not make any assumption about the sign of the
specific enthalpy.

Because of Eq.~\eqref{eq:depsilon}, Eq.~\eqref{eq:drhoInt} results in 
\begin{align} \label{eq:depsilonInt}
 &\frac{1}{\rho_1^{s_m^2}}\int_{\rho_1}^{~\rho} \rho'^{s_m^2} d\rho' 
 \le \frac{1}{h_1} \int_{\rho_1}^{~\rho} h d\rho' \le  \frac{1}{\rho_1}\int_{\rho_1}^{~\rho} \rho' d\rho'  \nonumber\\
 \Rightarrow &\frac{\rho_1}{1+{s_m^2}}\left[\left(\frac{\rho}{\rho_1}\right)^{s_m^2+1}-1\right] \le
 \frac{\epsilon-\epsilon_1}{h_1} \le \frac{\rho_1}{2}\left[\left(\frac{\rho}{\rho_1}\right)^{2}-1\right]\, ,
\end{align}
where $\displaystyle \int_{\rho_1}^{\rho} h d\rho'=\int_{\epsilon_1}^{\epsilon} d\epsilon'$ was employed.

From Eqs.~\eqref{eq:drho} and \eqref{eq:dpdh} we get
\begin{align} \label{eq:dpdrho}
 dp=c_s^2 h d\rho\, .
\end{align}
Taking into account Eq.~\eqref{eq:dpdrho}, from  Eq.~\eqref{eq:drhoInt} and
Eq.~\eqref{eq:SpSoL} we have 
\begin{align} \label{eq:dpInt}
 &\frac{s_m^2 }{\rho_1^{s_m^2}}\int_{\rho_1}^{~\rho} \rho'^{s_m^2} d\rho' 
 \le \frac{1}{h_1} \int_{\rho_1}^{~\rho} c_s^2  h d\rho' \le  \frac{1}{\rho_1}\int_{\rho_1}^{~\rho} \rho' d\rho'  \nonumber\\
 \Rightarrow &\frac{s_m^2 \rho_1}{1+{s_m^2}}\left[\left(\frac{\rho}{\rho_1}\right)^{s_m^2+1}-1\right] \le
 \frac{p-p_1}{h_1} \le \frac{\rho_1}{2}\left[\left(\frac{\rho}{\rho_1}\right)^{2}-1\right]\, ,
\end{align}
where $\displaystyle \int_{\rho_1}^{\rho} c_s^2 h d\rho'=\int_{p_1}^{p} dp'$ was employed.
Since $\rho \ge \rho_1$, inequality~\eqref{eq:dpInt} gives that
$(p-p_1)/h_1\ge 0$, while inequality~\eqref{eq:depsilonInt} gives that 
$(\epsilon-\epsilon_1)/h_1\ge 0$. For $\rho=\rho_1$,
Eqs.~\eqref{eq:drhoInt},~\eqref{eq:depsilonInt}, \eqref{eq:dpInt} reduce to
$h=h_1$, $\epsilon=\epsilon_1$, $p=p_1$ respectively, which is trivial but
self-consistent.   

\subsubsection{Assuming constant speed of sound}

Assuming $c_{s}^2 $ is independent of specific enthalpy, i.e. constant, then
by following similar steps as for arriving to the inequalities~\eqref{eq:drhoInt},~\eqref{eq:depsilonInt},~\eqref{eq:dpInt},
we get
\begin{align}
 \epsilon - \epsilon_1 &= \frac{1}{1+c_s^2}\, \rho_1\, h_1\, \left[\left( \frac{\rho}{\rho_1} \right)^{1+c_s^2} - 1 \right]\, ,\label{eq:epsilon} \\
 p - p_1 &= \frac{c_s^2}{1+c_s^2}\, \rho_1\, h_1\, \left[\left( \frac{\rho}{\rho_1} \right)^{1+c_s^2} - 1 \right]\label{eq:pressure}\, ,
\end{align}
which leads to
\begin{equation}\label{eq:eos}
 p  = c_s^2 \left(\epsilon - \epsilon_1 \right)\, +p_1\, .
\end{equation}
Note that if one changes the equation of the state of the fluid, i.e. $c_{s}^2 $,
the integration constants denoted with ``1'' change as well.

\section{Asymptotic behaviors}

\subsection{Rest-mass density}

The rest mass conservation~\eqref{eq:continuityeqn} can be rewritten as:
\begin{align} \label{eq:continuityeqn2}
 \dot{\rho}+\rho~\theta=0\, ,
\end{align}
where $\theta=\nabla_{\alpha}u^{\alpha}$ is the expansion scalar of the congruence
$u^{\alpha}$,  $\dot{\ }={u^\alpha }{\nabla _\alpha }$ denotes the derivative with
respect to a relevant time parameter $t$ along the congruence $u^{\alpha}$.
Integrating Eq.~\eqref{eq:continuityeqn2} along the time parameter $t$ leads to
\begin{equation}\label{eq:restmass_exp}
 \rho = \rho_0\ e^{-\int_{t_{0}}^t \theta(t')\, dt'}\, ,
\end{equation}
with initial condition $\rho(t_0)=\rho_0$.

\begin{prop} \label{prop:restmass}
For a perfect fluid moving along an expanding congruence 
with conserved positive rest-mass, the rest-mass density vanishes asymptotically, $\rho\rightarrow 0^{+}$, in the limit $t\rightarrow\infty$.
\end{prop}
\begin{proof}
 Since we have an expanding congruence, there exists a $k>0$, such
 that $\theta\ge k$. Eq.~\eqref{eq:restmass_exp} then leads to
\begin{align}
 \rho = \rho_0\ e^{-\int_{t_{0}}^t \theta(t')\, dt'} \le \rho_0\ e^{- \int_{t_{0}}^t k\, dt'}
 = \rho_0\ e^{-k(t-t_0)}\rightarrow 0\quad \text{for}\quad t\rightarrow\infty\, .
\end{align}
 Since $\rho > 0$, one has $\rho\rightarrow 0^{+} $ for $t\rightarrow\infty$,
 i.e. the rest mass density asymptotically vanishes.
\end{proof}

Proposition~\ref{prop:restmass} and the fact that $\rho_1\le\rho$ suggests that
$\rho_1$ must be an infinitesimally small positive quantity, i.e. $\rho_1\equiv 0^{+}$. 
Moreover, Proposition~\ref{prop:restmass} implies that for $t\rightarrow\infty$  
Eqs.~\eqref{eq:epsilon},~\eqref{eq:pressure} derived for a fluid with constant
non-zero speed of sound lead to
\begin{align}
 \epsilon - \epsilon_1 &\simeq -\frac{1}{1+c_s^2}\, \rho_1\, h_1\, , \label{eq:epsilonInf} \\
 p - p_1 &\simeq  -\frac{c_s^2}{1+c_s^2}\, \rho_1\, h_1\,  \label{eq:pressureInf}\, . 
\end{align}

To show an interesting implication of these relations, let us fix the constants of 
integration by considering the vanishing pressure limit, $p_1 =0$. In this limit,
one typically imposes that the specific enthalpy is equal to unity. Then, the relation $\epsilon+p=\rho h$, for  $p=p_1=0$ and $h=h_1=1$,
implies 
\begin{equation}
\epsilon_1 =\rho_1 \label{eq:enthalpyB}\, .
\end{equation}
With these constraints on the constants, we obtain the following expressions for
Eqs.~\eqref{eq:epsilonInf},~\eqref{eq:pressureInf}:
\begin{align}
 p &\simeq -\frac{\epsilon_1\, c_s^2}{1+c_s^2}\, ,\label{eq:pressBCS}\\
 \epsilon &\simeq \frac{\epsilon_1\, c_s^2}{1+c_s^2}\, .\label{eq:energyBCS}
\end{align}
It is immediately evident that Eq.~\eqref{eq:energyBCS} represents
a constant positive contribution to the energy density for any $c_s^2>0$,
if $\epsilon_1=\rho_1>0$. In a cosmological context such term behaves like a
{\it cosmological constant}, since $p=-\epsilon$. This has been already noticed for the case of the stiff fluid $(c_{\textrm{s}}=1)$ by \citet{Christodoulou:1995:ARRMA:}.

Applying proposition~\ref{prop:restmass} on the inequalities~\eqref{eq:depsilonInt},~\eqref{eq:dpInt}
and using the \eqref{eq:enthalpyB} choice for fixing the constants, we arrive at:
\begin{align}
 - \frac{\epsilon_1 s_m^2}{1+s_m^2} &\lesssim p \lesssim -\frac{\epsilon_1}{2}\, ,\label{eq:pressBIS}\\
  \frac{\epsilon_1 s_m^2}{1+s_m^2} &\lesssim \epsilon \lesssim \frac{\epsilon_1}{2}\, .\label{eq:energyBIS}
\end{align}
Eq.~\eqref{eq:energyBIS} still implies a constant positive contribution to the
energy density for $t\rightarrow\infty$, but Eq.~\eqref{eq:pressBIS} is only
possible if $\epsilon_1=0$, since $s_m^2\ll 1$. Thus, we are led to
$\epsilon_1=0$, which means that Eqs.~\eqref{eq:pressBIS},~\eqref{eq:energyBIS}
respectively lead to $p \simeq \epsilon \simeq 0$. Moreover, since the above
inequalities include the constant speed case as a subcase, then $\epsilon_1=0$
for Eqs.~\eqref{eq:pressBCS},~\eqref{eq:energyBCS}, so they do not imply the 
existence of a cosmological constant. On the other hand, this result might be
suggesting that the choice~\eqref{eq:enthalpyB} we have made to fix the constants
is not the proper one. 

In fact if we do not fix the constants, according to Proposition~\ref{prop:restmass}
the inequalities~\eqref{eq:drhoInt},~\eqref{eq:depsilonInt},~\eqref{eq:dpInt}
reduce to
\begin{align}
 \frac{h}{h_1} &\simeq 0 ,\label{eq:enthalpyBIA}\\
 -\frac{\rho_1}{1+{s_m^2}} &\lesssim \frac{\epsilon-\epsilon_1}{h_1} \lesssim -\frac{\rho_1}{2}\, , \label{eq:energyBIA}\\
 - \frac{s_m^2 \rho_1}{1+{s_m^2}} &\lesssim \frac{p-p_1}{h_1} \lesssim -\frac{\rho_1}{2}\, . \label{eq:pressBIA} 
\end{align}
Again because of $s_m^2\ll 1$, Eq.~\eqref{eq:pressBIA} can hold only if
$\rho_1$ is exactly zero. Note that even if $s_m^2$ was equal to zero
$\rho_1$ had to be zero as well. By not allowing the rest mass energy density to
acquire the zero value, we have arrived to a contradiction. If one would allow it,
then it would not be possible to derive the inequalities in Sec.~\ref{sec:ThermoRel}.
To resolve this contradiction, one might claim that the relations derived in
Sec.~\ref{sec:ThermoRel} hold only for finite time intervals, i.e. they do not hold
for $t\rightarrow \infty$. To discuss the asymptotic behaviors, we need propositions
like Proposition~\ref{prop:restmass}.

\subsection{Enthalpy}

Evaluating the thermodynamic relation Eq.~\eqref{eq:dh} along the flow lines, 
and implementing Eq.~\eqref{eq:adiabatic}, yields the relation
\be
{u^\alpha }{\nabla _\alpha }h = \frac{{hc_{\rm{s}}^2}}{\rho }{u^\alpha }{\nabla _\alpha }\rho\, ,
\ee
which can be used to rewrite the rest-mass conservation equation~\eqref{eq:continuityeqn}  as
\begin{align}
0 =& {\nabla _\alpha }(\rho {u^\alpha })\\ =& \frac{\rho }{{hc_{\rm{s}}^2}}({u^\alpha }{\nabla _\alpha }h + hc_{\rm{s}}^2{\nabla _\alpha }{u^\alpha })\, . \label{cont_h}
\end{align}
The continuity equation for the rest-mass density as expressed by Eq.~\eqref{cont_h} is
\be \label{hdot}
 \dot{h} = - c_{\rm{s}}^2\, \theta\, h\, ,
\ee
For generic time-dependent speed of sound and expansion scalar, one then has
\begin{equation}\label{enth_exp}
 h = h_0\ e^{-\int_{t_{0}}^t c_{\rm{s}}^2(t')\, \theta(t')\, dt'}\, ,
\end{equation}
with initial condition $h(t_0)=h_0$.

\subsubsection{Strong Energy Condition}

\begin{prop}\label{prop:enthalpySEC}
Consider a perfect fluid moving along an expanding and isotropic congruence, 
with conserved rest-mass and satisfying the Strong Energy Condition (SEC);
then if the speed of sound is a function of time defined in the interval $(0,1]$,
in the limit $t\rightarrow\infty$ one necessarily has $\epsilon\rightarrow0$ and
$p\rightarrow0$.
\end{prop}
\begin{proof}
The equation of rest-mass conservation can be rewritten in the form
Eq.~\eqref{hdot}, whose general solution is given by eq.\eqref{enth_exp}.
We would like to evaluate the behavior of $h$ in the limit when
$t\rightarrow\infty$ by obtaining an upper and a lower bound. 

{\it Lower bound.}  First of all $c_s^2(t)\in (0,1]$, so we can write
\be\label{eq:enth2}
 h = h_0\ e^{-\int_{t_{0}}^t c_{\rm{s}}^2(t')\, \theta(t')\, dt'}\ \geq\  h_0\ e^{-\int_{t_{0}}^t \theta(t')\, dt'} \, .
\ee
Secondly, the Raychaudhuri equation for an isotropic timelike congruence $u^{\alpha}$ reads
\be \label{eq:RayIsot}
 \dot{\theta} = - \left(\frac{1}{3}\theta^2 + R_{\alpha\beta}u^{\alpha}u^{\beta}\right)\, .
\ee
Because of the SEC, the last term is positive.  Hence we get the inequality
\be
 \dot{\theta} \leq - \frac{1}{3}\theta^2\, .
\ee
Integration of such inequality gives
\be
 \theta \leq \frac{3\, \theta_0}{3+\theta_0\, t}\, ,
\ee
with $\theta_0=\theta(t_0)$.  Applying such bound to the rightmost term of
Eq.~\eqref{eq:enth2} gives
\begin{align}
 h \ge h_0\ e^{-\int_{t_{0}}^t \theta(t')\, dt'} &\geq h_0\, e^{-\int_{t_0}^t\frac{3\, \theta_0}{3+\theta_0\, t'}\, dt'}\\ 
 &= h_0\, \left( \frac{3+\theta_0 t_0}{3+\theta_0 t} \right)^3 \rightarrow 0\quad \text{for}\quad t\rightarrow\infty\, .
\end{align}
Hence $h\geq0$ for $t\rightarrow\infty$.

{\it Upper bound.}  By assumption, the product $c_s^2(t) \theta(t)$ is strictly
positive: hence there exist a constant $k>0$ such that $c_s^2(t) \theta(t)\geq k>0$
for any finite time.  The function $h$ can then be bounded from above in the
following way:
\begin{align}
 h = h_0\ e^{-\int_{t_{0}}^t c_{\rm{s}}^2(t')\, \theta(t')\, dt'} &\leq h_0\ e^{- \int_{t_{0}}^t k\, dt'}\nonumber\\
 &= h_0\ e^{-k(t-t_0)}\rightarrow 0\quad \text{for}\quad t\rightarrow\infty\, .
\end{align}
Hence $h\leq0$ for $t\rightarrow\infty$.

Putting together the results of both bounds, we find that $h=0$ in the limit
$t\rightarrow\infty$.  At the same time $\rho\rightarrow0$ in the same limit,
because of Proposition~\ref{prop:restmass}.  Thus, one has that
$h\equiv\frac{\epsilon+p}{\rho}\rightarrow0$ implies that $p+\epsilon\rightarrow0$.  

Lastly, the SEC requires $p+\frac{1}{3}\epsilon \geq 0$: the only case in which
the condition $p+\epsilon\rightarrow0$ is consistent with this bound is when both
$\epsilon\rightarrow0$ and $p\rightarrow0$ (left panel of Fig.~\ref{fig:examples}).
\end{proof}

Note that Proposition~\ref{prop:restmass} by itself could not lead to
$p+\epsilon\rightarrow0$, since the asymptotic bounded value of the specific
enthalpy was not guaranteed.

{\it Proposition~\ref{prop:enthalpySEC} is a general statement about the 
impossibility for a ``well defined'' isotropic perfect fluid satisfying the
SEC to have a non-trivial pressure asymptotically.} Hence, in the following 
prepositions we drop SEC and specialize to a spatially flat Friedmann-Robertson-Walker (FRW) spacetime.

\subsubsection{Bounded Rate of Expansion}

\begin{prop}\label{prop:enthalpyBExp}
Consider a perfect fluid moving along an expanding congruence in a flat FRW spacetime, 
with conserved rest-mass and a rate of expansion bounded by $\Xi$; then
if the speed of sound is a function of time defined in the interval $(0,1]$, in
the limit $t\rightarrow\infty$ one has $\epsilon+p\rightarrow0$, without necessarily
$\epsilon\rightarrow0$ and $p\rightarrow0$, and $0 \lesssim \Xi$.
\end{prop}

\begin{proof}
 The {\it upper bound} stays the same as in Proposition~\ref{prop:enthalpySEC},
 so $h\leq0$ for $t\rightarrow\infty$.
 {\it Lower bound.} A positive, but bounded rate of congruence expansion means that
 $\dot{\theta}\le \Xi$, thus $\theta(t)\le \Xi(t-t_0)+\theta_0$. Then,
 Eq.~\eqref{eq:enth2} gives
 \begin{align}
 h \ge h_0\ e^{-\int_{t_{0}}^t \theta(t')\, dt'} &\geq h_0\, e^{-\int_{t_0}^t \Xi(t-t_0)+\theta_0 dt'}\nonumber\\ 
 &= h_0\, e^{-(\Xi(t-t_0)^2/2+\theta_0 (t-t_0))} \rightarrow 0\quad \text{for}\quad t\rightarrow\infty\, .
\end{align}
Putting together the results of both bounds, we find that $h=0$ in the limit
$t\rightarrow\infty$.  Thus, again one has that $p+\epsilon\rightarrow0$.

However, from the isotropic Raychaudhuri Eq.~\eqref{eq:RayIsot} we have:
\be
 - \left(\frac{1}{3}\theta^2 + R_{\alpha\beta}u^{\alpha}u^{\beta}\right)\le \Xi \Rightarrow
 -\frac{3}{2} (\epsilon+ p)\le \Xi\, ,
\ee
where we used Friedmann equation $\theta^2=3\epsilon$. Thus, in this case the solution $\epsilon\rightarrow0$, $p\rightarrow0$ is not the only allowed to have  
$\epsilon+p\rightarrow0$  (right panel of Fig.~\ref{fig:examples}). Actually, $p\rightarrow -\epsilon$ implies that $0\lesssim \Xi$. 
\end{proof}

Note that proposition~\ref{prop:enthalpyBExp} allows an exponential growth for FRW
$$3 \frac{\dot{a}}{a}= \theta= \Xi (t-t_0)+\theta_0 \Rightarrow a\le a_0 e^{(\Xi (t-t_0)^2/2+\theta_0(t-t_0))/3}$$ even if $\Xi=0$. Thus, to have exponential growth the minimal requirement
is that $\dot{\theta}\le 0$.

\begin{figure} [t]
\begin{center}
\includegraphics[width=0.25\linewidth,height=0.5\linewidth]{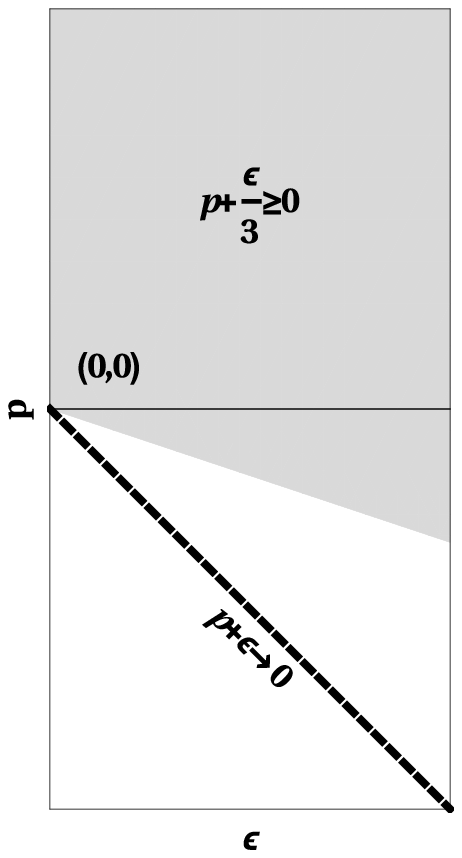}
\includegraphics[width=0.25\linewidth,height=0.5\linewidth]{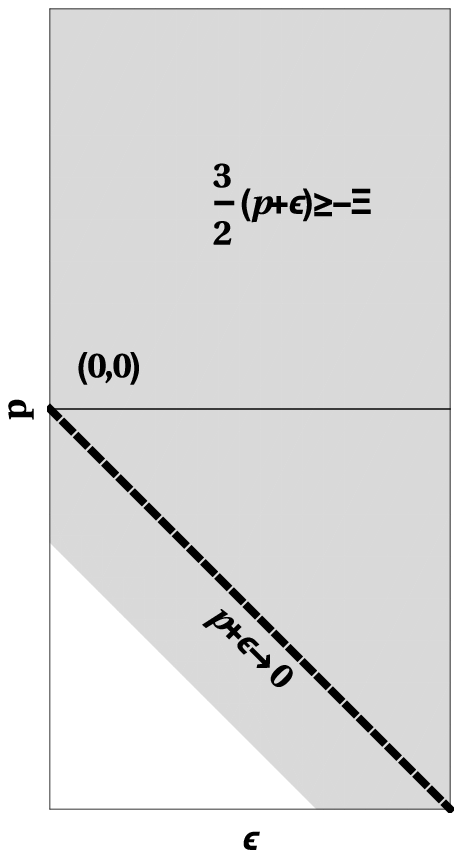}    
\end{center}
% The package `graphicx' unconditionally loaded by the document class
% provides standard tools for picture inclusion and many other useful
% things.  However, knowing the \includegraphics command is sufficient
% for most writing.  The aspect ratio of the following picture is slightly
% deformed by specifying both the width and height to fit the page.
% Normally use either of these, but not both. It's good to specify
% the dimensions in fractions of page layout sizes, e.g., the text width
% stored in \linewidth.  Not specifying the graphics file extension
% (.ps, .eps, or .pdf) is a good idea as well.
\caption{\label{fig:examples}
Left panel: The plane of allowed EoS assuming SEC, Proposition~\ref{prop:enthalpySEC}. Right Panel: The plane of allowed EoS assuming bounded rate of congruence expansion, Proposition~\ref{prop:enthalpyBExp}. In both panels we assume that the energy density is $\epsilon\ge 0$.}
\end{figure}

\section{Summary}

Starting from a general thermodynamical treatment of usual matter, in the form of an irrotational perfect fluid, our investigation indicates that a constant speed of sound for usual matter is not a viable way to provide a cosmological constant.  We have given a formal proof that if the strong energy condition holds, usual matter cannot provide negative pressure.  Moreover, we have provided a formal proof that for a flat FRW spacetime containing only usual matter, for which the strong energy condition is violated, negative pressure is possible .

\ack%%%%%%%%%%%%%%%%%%%%%%%%%%%%%%%%%%%%%%%%%%%%%%%%%%%%%%%%%%%%%%%%%%%%%%%

% Contributors involved in `Vyzkumny zamer' can use macro \InstResCode
% instead of specifying the alphanumerical code explicitly:
G.L-G is supported by Grant No. GA\v{C}R-17-06962Y of the Czech Science
Foundation.  G.A. is supported by Grant No. GA\v{C}R-17-16260Y of the Czech Science Foundation.
C.M. is supported by the European Union's Horizon 2020
research and innovation programme under the Marie Skłodowska-Curie grant agreement No 753115.
% Here we specify the basename of the bibliography database file,
% in this case \jobname=ragsamp:
\bibliography{RAGtime20_ggl}
\end{document}